\documentclass[journal,11pt,draftcls,onecolumn]{IEEEtran} 

\usepackage{amsmath,amssymb,amsthm}
\usepackage{enumerate}
\usepackage{stfloats}
\usepackage{comment}
\usepackage{subfig}

\usepackage[T1]{fontenc}
\usepackage{graphics} 
\usepackage{epsfig} 
\usepackage[mathscr]{euscript}
\usepackage{algorithm}
\usepackage[noend]{algpseudocode}
\makeatletter
\def\BState{\State\hskip-\ALG@thistlm}
\makeatother

\usepackage{tikz}
\usetikzlibrary{arrows,shapes,chains,matrix,positioning,scopes,patterns,calc}
\usepackage{pgfplots}
\pgfplotsset{compat=1.3}
\usepgflibrary{shapes}

\newtheorem{theorem}{Theorem}
\newtheorem{lemma}[theorem]{Lemma}

\newtheorem{definition}[theorem]{Definition}

\renewcommand{\epsilon}{\varepsilon}

\newcommand\indep{\protect\mathpalette{\protect\independenT}{\perp}}
\def\independenT#1#2{\mathrel{\rlap{$#1#2$}\mkern2mu{#1#2}}}

\newcommand{\mb}[1]{\mathbf{#1}}
\newcommand{\mbb}[1]{\mathbb{#1}}

\newcommand{\mc}[1]{\mathcal{#1}}

\newcommand{\msc}[1]{\mathscr{#1}}

\newcommand{\RNum}[1]{\uppercase\expandafter{\romannumeral #1\relax}}

\newcommand{\bv}{\vec{\mathrm{b}}}
\newcommand{\xv}{\vec{\mathrm{x}}}
\newcommand{\yv}{\vec{\mathrm{y}}}
\newcommand{\zv}{\vec{\mathrm{z}}}
\newcommand{\rv}{\vec{\mathrm{r}}}
\newcommand{\wv}{\vec{\mathrm{w}}}

\newcommand{\Xv}{\vec{X}}
\newcommand{\Yv}{\vec{Y}}

\newcommand{\Rv}{\vec{R}}

\DeclareMathAlphabet{\mcl}{OMS}{cmsy}{m}{n}

\DeclareMathOperator*{\argmax}{\,arg\ max}

\newlength\tikzwidth
\newlength\tikzheight

\newcommand{\coleq}{\mathrel{\mathop:}=}
\newcommand{\defeq}{\triangleq}

\begin{document}
\title{Sub-string/Pattern Matching in Sub-linear Time Using a Sparse Fourier Transform Approach}

\author{Nagaraj T. Janakiraman, Avinash Vem, Krishna R. Narayanan, Jean-Francois Chamberland\\
Department of Electrical and Computer Engineering \\
Texas A\&M University\\
{\tt\small {\{tjnagaraj,vemavinash,krn,chmbrlnd\}@tamu.edu} }}
	\maketitle

\begin{abstract}
	We consider the problem of querying a string (or, a database) of length $N$ bits to determine all the locations where a substring (query) of length $M$ appears either exactly or is within a Hamming distance of $K$ from the query. We assume that sketches of the original signal can be computed off line and stored. Using the sparse Fourier transform computation based approach introduced by Pawar and Ramchandran, we show that all such matches can be determined with high probability in sub-linear time. Specifically, if the query length $M = O(N^\mu)$ and the number of matches $L=O(N^\lambda)$, we show that for $\lambda < 1-\mu$ all the matching positions can be determined with a probability that approaches 1 as $N \rightarrow \infty$ for $K \leq \frac{1}{6}M$. More importantly our scheme has a worst-case computational complexity that is only $O\left(\max\{N^{1-\mu}\log^2 N, N^{\mu+\lambda}\log N \}\right)$, which means we can recover all the matching positions in {\it sub-linear} time for $\lambda<1-\mu$. This is a substantial improvement over the best known computational complexity of $O\left(N^{1-0.359 \mu} \right)$ for recovering one matching position by Andoni {\em et al.} \cite{andoni2013shift}. Further, the number of Fourier transform coefficients that need to be computed, stored and accessed, i.e., the sketching complexity of this algorithm is only $O\left(N^{1-\mu}\log N\right)$. Several extensions of the main theme are also discussed.
\end{abstract}



\section{Introduction and Problem Statement}
\label{sec:introduction}
We consider the substring/pattern matching problem, which has been studied extensively in theoretical computer science.
In this problem, a signal $\xv \coleq (x[0], x[1], \ldots, x[N-1])$ of length $N$ symbols representing a string, library, or database is available.
The objective is to answer queries regarding whether a given string $\yv \coleq (y[0], y[1],\ldots, y[M-1])$ of length $M$ is a substring of $\xv$.
We are especially interested in the case where a sketch of $\xv$ can be computed offline and stored, and where the one-time computational complexity of creating the sketch can be amortized over repeated queries or ignored.
Moreover, we focus on the random setting in which the $x[i]$'s form a sequence of independent and identically distributed (i.i.d.) random variables, each taking values in $\mathcal{A} \subset \mathbb{R}$. We begin our analysis by restricting our attention to the binary case where $\mathcal{A} \coleq \{\pm 1\}$.\footnote{Extensions to other alphabets $\mathcal{A} \subset \mathbb{R}$ are straightforward. Extension to the non-i.i.d.\ case is also possible.}
Within this context, we consider the following two settings:

{\bf Exact Pattern Matching:} In the exact pattern matching problem, a substring of length $M$ of $\xv$, namely $\yv \coleq \xv[\tau:\tau+M-1]$, is obtained by taking $M$ consecutive symbols from $\xv$ and is presented as a query. This pattern may appear within $\xv$ in up to $L$ different locations $\tau_1, \tau_2, \ldots, \tau_L$ and the task is to determine all the locations $\tau_i, \forall i\in[1:L]$.
We consider the probabilistic version where our objective is to recover the matching locations with a probability that approaches 1 as $M,N \rightarrow \infty$.

{\bf Approximate Pattern Matching:} In approximate pattern matching, $\yv$ is a noisy version of a substring, i.e., $\yv = \xv[\tau:\tau+L-1] \odot \bv$, where $\bv$ is a noise sequence with $b[i] \in \{ \pm 1 \}$ such that $d_H(\yv,\xv[\tau:\tau+M-1]) \leq K$.
Function $d_H$ denotes the Hamming distance, and $\odot$ represents component-wise multiplication. The objective is to determine all locations $\tau_i$ such that $d_H(\yv,\xv[\tau_i:\tau_i+M-1]) \leq K$ with a probability that approaches 1 as $K, M$ and $N \rightarrow \infty$.

We evaluate our proposed algorithm according to two metrics - (i) the space required to store the sketch of $\xv$, which we refer to as sketching complexity, and (ii) the computational complexity in answering the query.

These problems are relevant to many applications including text matching, audio/image matching, DNA matching in genomics, metabolomics, radio astronomy, searching for signatures of events within large databases, etc. The proposed techniques are particularly relevant now due to the interest in applications involving huge volumes of data. Our proposed approach is most useful in the following situations.
(i) The string $\xv$ is available before querying and one time computations such as computing a sketch of $\xv$ can be performed offline and stored, and the complexity of computing the sketch of $\xv$ can be ignored. Then, when queries in the form of $\yv$ appear, one would like to decrease the computational complexity in searching for the string $\yv$. (ii) The string $\xv$ is not centrally available, but parts of the string are sensed by different data collecting nodes distributively and communicated to a central server. A search query is presented to the server; and this server must decide whether the string appears in the data sensed by one or more of the distributed nodes and, if present, it must also identify when the queried string appeared.
In this latter case, we wish to minimize the amount of data communicated by the nodes to the server and the computational complexity in searching for the string.
The proposed formulation is most useful when the query $\yv$ is not a pattern that can be predicted and, therefore, creating a lookup table to quickly identify commonly-occurring patterns will not be effective.

A naive approach to searching for substring $\yv$ in $\xv$ is to first compute the cross-correlation between $\xv$ and $\yv$, which we denoted by $\rv=[r[0], r[1],\ldots, r[N-1]]$ with
\begin{align}
\label{Eqn:DefCrossCorrelation}
r[m]=(\xv*\yv)[m] \defeq \sum_{i=1}^{M} x[m+i-1] y[i] ,
\end{align}
and, subsequently,  choosing the index $m$ that maximizes $r[m]$.
This strategy uses the $N$ samples of $\xv$ and has a super-linear computational complexity of $MN = O(N^{1+\mu})$. A computationally more efficient approach uses the fact that $\rv$ can be computed by taking the inverse Fourier transform of the product of the Fourier transforms of $\xv$ and $ \yv^*[-n]$, where $\yv^*[-n]$ is the conjugate of the time reversed version of $\yv$.
Such an approach still uses all the $N$ samples of $\xv$, but reduces the computational complexity to $O(N \log N)$. Note that even though the Fourier transform of $\xv$ can be precomputed, the $N$-point Fourier transform of $\yv$ still needs to be computed online resulting in the $O(N \log N)$ computational complexity.

Both the exact pattern matching problem and the approximate pattern matching problem have been extensively studied in computer science.
Three recent articles \cite{andoni2013shift}, \cite{amir2004faster} and \cite{navarro2001guided} provide a brief summary of existing contributions.
For the exact matching problem, the Rabin-Karp algorithm solves a more general problem of finding the occurrence of a set of query strings.
However, the algorithm has a computational complexity that is at least linear in $N$. Boyer and Moore presented an algorithm in \cite{boyer1977fast} for finding the first occurrence of the match (only $\tau_1$) that has an expected complexity of $O(N/M \log N) = O(N^{1-\mu} \log N)$, whereas the worst case complexity (depending on $\tau_1$'s) can be $O(N \log N)$. For large $M$, the algorithm indeed has an average complexity that is sub-linear in $N$. More recently, it has been shown that techniques based on the Burrows-Wheeler transform can be used to solve the exact matching problem with sub-linear time complexity \cite{ferragina2005indexing} using a storage space of $O(N)$ bits. This problem is well studied under the read alignment setting by the Bioinformatics community \cite{li2009fast,li2010fast}. For small $|\mathcal{A}|$, it has the best known complexity; however, the complexity increases with $|\mathcal{A}|$. Further, extensions to approximate matching setting \cite{zhang2003approximate} has a complexity that increases exponentially in $K$ and, hence, appear to be infeasible for $K = O(M)$. The Boyer and Moore algorithm has been generalized to the approximate pattern matching problem in \cite{chang1994approximate} with an average case complexity of $O(NK/M \log N)$, which provides a sub-linear time algorithm only when $K \ll M$. In \cite{andoni2013shift}, Andoni, Hassanieh, Indyk and Katabi have given the first sub-linear time algorithm with a complexity of $O(N/M^{0.359})$ for $K = O(M)$.

\section{Our Main Results and Relation to Prior Work}
\label{sec:mainresults}

	Assume that a sketch of $\xv$ of size $O(\frac{N}{M} \log N)$ can be computed and stored. Then for the {\it exact pattern matching} and {\it approximate pattern matching} (with $K = \eta M$) problems, with the number of matches $L$ scaling as $O(N^{\lambda})$, we show an algorithm that has
	\begin{itemize}
		\item a sketching function for $\yv$ that computes $O(\frac{N}{M}\log N)~=~ O\left(N^{1-\mu}\log N\right)$ samples
		\item a computational complexity of $O\left(\max\{N^{1-\mu}\log^2 N, N^{\mu+\lambda}\log N \}\right)$
		\item a decoder that recovers all the $L$ matching positions with a failure probability that approaches zero asymptotically in $N$		
	\end{itemize} 
	When $L<O\left(\frac{N}{M}\right)$ (i.e. $\lambda<1-\mu$), which is typically the interesting case, our algorithm has a {\it sub-linear time and space complexity}.


There are important differences between our paper and \cite{hassanieh2012faster,andoni2013shift,boyer1977fast,amir2004faster}. First and foremost, the algorithms used for pattern matching are entirely different. While their algorithms are combinatorial in nature, our algorithm is algebraic and uses signal processing and coding theoretic primitives. Secondly, the system model considered in our paper differs from the model in \cite{hassanieh2012faster,andoni2013shift,boyer1977fast,amir2004faster} in that we allow for preprocessing or creating a sketch of the data $\xv$. Our algorithm exploits this fact and results in a computational complexity $O(N/M)$ which is better than that in \cite{andoni2013shift} for the approximate pattern matching problem.  Finally, we also consider the problem of finding all matches of the pattern $\yv$ instead of looking for only one match.

Our paper is inspired by and builds on two recent works by Hassanieh {\em {et al.}} in \cite{hassanieh2012faster} and Pawar and Ramchandran \cite{pawar2014robust}. In \cite{hassanieh2012faster}, Hassanieh \emph{et al.}, considered the correlation function computation problem for a Global Positioning System (GPS) receiver and exploited the fact that the cross-correlation vector $\rv$ is a very sparse signal in the time domain and, hence, the Fourier transform of $\rv$ need not be evaluated at all the $N$ points. In the GPS application, which was the focus of \cite{hassanieh2012faster}, the query $\yv$ corresponds to the received signal from the satellites and, hence, the length of the query was at least $N$. As a result, the computational complexity is still $O(N \log N)$ (still linear in $N$) and only the constants were improved in relation to the approach of computing the entire Fourier transform. In a later paper by Andoni {\em et al.,} \cite{andoni2013shift}, a sub-linear time algorithm for shift finding in GPS is presented; however, this algorithm is completely combinatorial and eschews algebraic techniques such as FFT-based techniques.

In \cite{pawar2014robust}, Pawar and Ramchandran present an algorithm based on aliasing and the peeling decoder for computing the Fourier transform of a signal with noisy observations for the case when the Fourier transform is known to be sparse. This algorithm has a complexity of $O(N \log N)$ and they do not consider the pattern matching problem. Our algorithm is similar to that of Pawar and Ramchandran's algorithm with one important distinction. We modify their algorithm to exploit the fact that the peak of the correlation function of interest is always positive. This modification allows us to compute the Sparse Inverse Discrete Fourier Transform (SIDFT) with sub-linear time complexity of $O(N^{1-\alpha} \log N)$, $0 < \alpha \leq 1$ . One of the main contributions of this paper is to show that signal processing and coding theoretic primitives, i.e., Pawar and Ramchandran's algorithm with some modifications can be used to solve the pattern matching algorithm in sub-linear time.

\section{Notations}

The table below introduces the notations we adopt throughout this paper.
\begin{table}[h!]
\begin{center}
	\begin{tabular}{cc} 	
		\hline		
		\textit{Symbol}	    &  \textit{Meaning} \\		
		\hline
		$N$           				& Size of the string or database in symbols \\
	
		$M$   				        & Length of the query in symbols \\

        $L$    						&   Number of matches \\
		$\mu$ 				        & Smallest $0<\mu<1$ such that $M =O(N^{\mu})$\\		
		$\lambda$       		& Smallest $0<\lambda<1$ such that $L =ON^\lambda)$\\
        $K$                        &$\max_{\tau}d_{H}(\xv[\tau:\tau+M-1],\yv)$\\
	    $\eta$             &$\frac{K}{M}$\\
$d$           				& Number of stages in the FFAST algorithm \\
$f_i \approx N^\alpha$     & Length of smaller point IDFT at each stage-$i$\\
$g_i = N/f_i$     	    &  Sub-sampling factor in Fourier domain for stage-$i$\\
$B$   					    & Number of shifts (or branches) in each stage \\
$G = N^\gamma$    & Number of blocks (for parallel processing)\\
$\tilde{N} = N^{1-\gamma}$   & Length of one block (for parallel processing)\\
		\hline
	\end{tabular}
\end{center}	
\caption{Parameters and various quantities involved in describing the algorithm}
\label{Table:Notations}
\end{table}	
We denote signals and vectors using the standard vector notation of arrow over the letter, time domain signals using lowercase letters and the frequency domain signals using uppercase letters. For example $\xv = (x[0], x[1], \ldots, x[N-1] )$ denotes a time domain signal with $i^{\text{th}}$ time component denoted by $x[i]$, and $\Xv= \mathcal{F}_{N}\{\xv\}$ denotes the $N$-point Fourier coefficients of $\xv$. We denote matrices using boldface upper case letters. We denote the set $\{0,1,2\cdots, N-1\}$ by $[N]$.
\section{Description Of The Algorithm}
\label{sec:Algo_desc}
In this section, given the input string $\xv$ and the query string $\yv$, we describe our algorithm that finds the matching positions $\mathcal{T}\coleq \{\tau_1, \tau_2, \cdots \tau_L\}$ with sample and time complexities that are sub-linear in $N$. The main idea exploits the fact that the correlation vector $\rv$ is sparse (upto some noise terms) with dominant peaks at $L$ matching positions denoted by $\mc{T}$ and noise components at $N-L$ positions where the strings do not match.
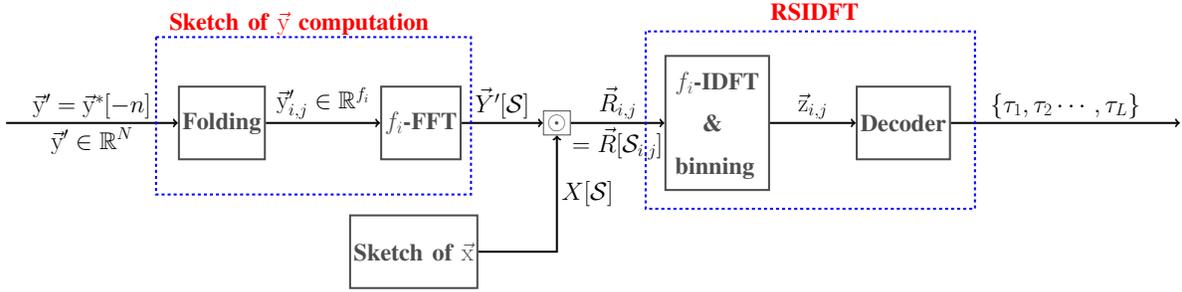
\begin{figure*}[t!]
  \centering
	 	\resizebox{0.95\textwidth}{!}{\begin{tikzpicture}

\def\nodewidth{1in}
\def\fsize{\Huge}
\tikzstyle{block} = [rectangle, draw, thick,opacity=0.7,line width =2, minimum size=\nodewidth]
\tikzstyle{opnode} = [rectangle, draw, thick,opacity=0.7,line width=1, minimum size=0.2in]

\node[block] (r1) at (-8.7,3){\fsize \bf Folding};
\node[block] (r2) at (-1.75,3) {\fsize \bf $f_i$-FFT};
\node[opnode] (r3) at (3,3) {\fsize \bf $\odot$};
\node[block,align=center] (r4) at (8.6,3) {\fsize \bf \begin{tabular}{l}
$f_i$-IDFT \\ 
~~ \& \\ 
binning
\end{tabular}};
\node[block] (r5) at (15.1,3) {\fsize \bf Decoder};
\node[block] (r6) at (-2,-1.5) {\fsize \bf Sketch of $\xv$};

\draw[<-, thick, line width=2] (r1.west)--node[midway, above]{\fsize $\yv'=\yv^*[-n]$}
node[midway, below]{\fsize $\yv'\in\mathbb{R}^N$}+(-6,0);
\draw[->, thick, line width=2] (r1.east)--node[midway, above]{\fsize $\yv'_{i,j}\in\mathbb{R}^{f_i}$}(r2.west);
\draw[->, thick, line width=2] (r2.east)--node[midway, above]{\fsize $\Yv'[\mc{S}]$}(r3.west);
\draw[->, thick, line width=2] (r3.east)--node[midway, above]{\fsize $\Rv_{i,j}$}
node[midway,below]{\fsize \bf $=\Rv[\mathcal{S}_{i,j}]$}(r4.west);
\draw[->, thick, line width=2] (r4.east)--node[midway, above]{\fsize $\zv_{i,j}$}(r5.west);
\draw[->, thick, line width=2] (r5.east)--node[midway, above]{\fsize $\{\tau_1, \tau_2 \cdots, \tau_L \}$}+(8,0);
\draw[->, thick, line width=2] let \p1=(r6),\p2= (r3) in (r6.east)--(\x2,\y1)--node[midway, right]{\fsize \bf $X[\mathcal{S}]$}  (r3.south);

\draw [dashed, line width =2, color=blue] (-11,6) rectangle (0,0.5);
\node[color= red] at (-6,6.5) {\fsize \bf  Sketch of $\yv$ computation};
\draw [dashed, line width =2, color=blue]  (6.1,6.2) rectangle (17.6,0);
\node [color = red] at (12,6.9) {\fsize \bf RSIDFT};
\end{tikzpicture}}	
	\caption{Schematic of the proposed scheme using sparse Fourier transform computation.}\label{fig:notional}
\end{figure*}

Consider the correlation signal $\rv$ in the case of exact matching:
\begin{equation} \label{eqn:RXY_sparse}
r[m] \ = \left\{
\begin{array}{ll}
  &M,~~  \text{if} \ m \in \mathcal{T} \\
  & n_m,~~ m \in [N]-\mathcal{T}
\end{array}
\right.
\end{equation}
where $n_m$ is the noise component that is induced due to correlation of two i.i.d. sequence of random variables each taking values from $\mathcal{A} := \{+1,-1\}$. The sparse vector $\rv$ can be computed indirectly using Fourier transform approach as shown below:
\begin{equation}\label{eqn:Rxy_fourier}
  \rv = \underset{\text{ \RNum{3} } } {\mathcal{F}_{N}^{-1}} \ \{ \underset{\text{ \RNum{1} } }{  \mathcal{F}_{N}\{\xv\}}  \odot \ \underset{\text{ \RNum{2} } }{ \mathcal{F}_{N}\{\yv'\}}  \}
\end{equation}
where $\mathcal{F}_{N}\{ \cdot \}$ and $\mathcal{F}_{N}^{-1}\{ \cdot \}$ refer to $N$-point discrete Fourier transform and its inverse respectively, $\odot$ is the point-wise multiplication operation and ${ y'[n]} = { y^{*}[-n]}$. Fig.~\ref{fig:notional} presents a notional schematic of our Algorithm. As evident from Eq.~\eqref{eqn:Rxy_fourier}, our algorithm for computing $\rv$ consists of three stages:
\begin{itemize}
\item Computing the sketch $\Xv=\mathcal{F}_{N}\{\xv\}$ of $\xv$
\item Computing the sketch $\Yv'=\mathcal{F}_{N}\{\yv'\}$ of $\yv$
\item Computing the IDFT of $\Rv=\Xv \odot \Yv'$ given $\Xv$ and $\Yv'$
\end{itemize}

\subsection{Sparse Inverse Discrete Fourier Transform}
\label{subsec:RSIDFT}	
 In this section we present Robust Sparse Inverse Discrete Fourier Transform(RSIDFT) scheme that exploits sparsity in the cross-correlation signal $\rv$ and efficiently recovers its $L$ dominant coefficients. The architecture of RSIDFT is similar to that of the FFAST scheme proposed in \cite{pawar2014robust}, but the decoding algorithm has some modifications to handle the noise model induced in this problem. We will see in Sec.~\ref{subsec:skteches} how the sketches $\Xv$ and $\Yv'$ are computed efficiently but for this section we will focus only on the recovery of the sparse coefficients in $\rv$ given $\Xv$ and $\Yv'$.
	  	
Consider the RSIDFT framework shown in Fig.~\ref{fig:rsidft}. Let $ \Rv =\Xv \odot \Yv'$ be the DFT of the cross-correlation signal of $\xv$ and $\yv$. We begin by factoring $N$ into $d$ relatively prime factors $\{f_1,f_2,\ldots,f_d\}$, where $d$ is a parameter in the algorithm. The design scheme for choosing $f_i$'s for various values of $\mu$ such that $f_i$ divides $N$ and $f_i=N^{\alpha}+O(1) \forall ~i\in[d]$ are given in Sec. \ref{subsec:DesignParameters}. The RSIDFT algorithm consists of {\it $d$-stages} with each stage corresponding to a sub-sampling factor of $\frac{N}{f_i}$. In each stage, there are {\it $B= O(\log N)$ branches} with shifts from the set $ \{s_1, s_2, \cdots s_B\} $, where $s_1 =0$ in the first branch and the rest are chosen uniformly at random from $[N]$.
	   	 	
\begin{figure}[h!]
	\begin{center}
	 	\includegraphics[scale=0.9]{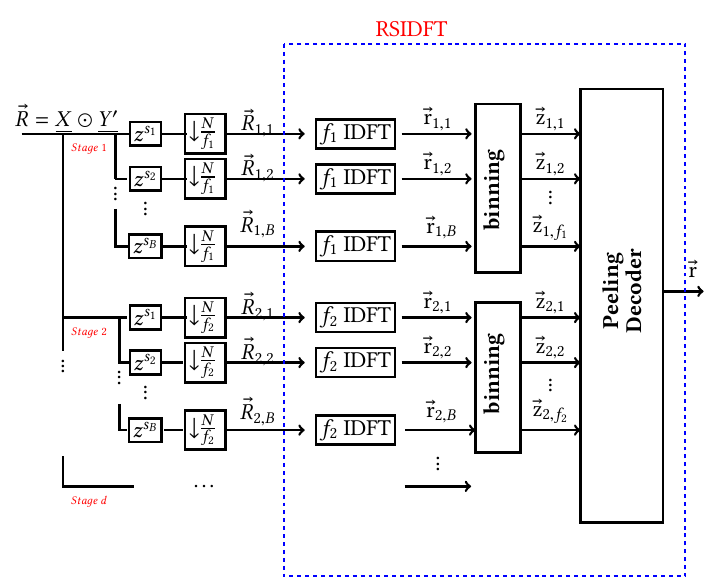}
	\end{center}	
\caption{ RSIDFT Framework to compute inverse Fourier Transform of a signal $\Rv$ that is sparse in time domain. }\label{fig:rsidft}
\end{figure}

	 Given the input $\Rv$, in branch $j$ of $i^{\text{th}}$ stage of RSDIFT, referred to as \textit{branch $(i,j)$} for simplicity, RSIDFT sub-samples the signal $\Rv$ at
\begin{align}
\label{Eqn:SamplingSets}
	 \mathcal{S}_{i,j} \coleq \{s_j,\ s_j + g_i,\ s_j + 2g_i,\ \cdots s_j + (f_i-1) g_i\}, \quad i\in[d], j\in[B]
\end{align}
where $g_i\coleq \frac{N}{f_i}$ to obtain $\Rv_{i,j}\coleq\Rv[\mc{S}_{i,j}]$. The sub-sampling operation is followed by a $f_i$-point IDFT in \textit{branch (i,j)} of stage-$i$ to obtain $ \rv_{i,j}$. Notice that $ \rv_{i,j}$ is an aliased version of $\rv$ due to the property that sub-sampling in Fourier domain is equivalent to aliasing in time domain.

	 Let $\zv_{i,k}\in\mbb{R}^{B}$, for $k\in [f_i]+1$, be the $k$th \textit{binned} observation vector of stage-$i$ formed by stacking $\rv_{i,j}[k], j\in[B]$, together as a vector i.e.
\[
	  \zv_{i,k} = \begin{bmatrix}
	 r_{i,1}[k],~ 
	 r_{i,2}[k],~ 
	 \cdots,~
	 r_{i,B}[k]
	 \end{bmatrix}^{T}
\]
Note that this gives us a total of $f_i$ binned observation vectors in each stage-$i$. Using the properties of Fourier transform, we can write the relationship between the observation vectors $\zv_{i,k}$ at bin $(i,k)$ and sparse vector to be estimated $\rv$ as: 
\begin{align}
	\label{Eqn:Generator Equation}
	\zv_{i,k}= \mb{W}_{i,k} \times
	\begin{bmatrix}
		r[k],~
		r[k+f_i],~
		\cdots,~
		r[k+(g_i-1)f_i]
	\end{bmatrix}^{T}
\end{align}

where we refer to $\mb{W}_{i,k}$ as the sensing matrix at bin $(i,k)$ and is defined as
\begin{align}\label{Eqn:Sensing Matrix}
	\mb{W}_{i,k} = \left[\wv^{k},\wv^{k+f_i},\ldots, \wv^{k+(g_i-1)f_i}\right]
\end{align} and
 $\wv^{k}=
	\begin{bmatrix}
		e^{\frac{j2\pi ks_1}{N}}, ~
		e^{\frac{j2\pi ks_2}{N}},~
		\cdots,~
		e^{\frac{j2\pi ks_B}{N}}
	\end{bmatrix}^{T}$

\begin{figure}[h!]
	\begin{center}
	 	\resizebox{0.70\textwidth}{!}{\begin{tikzpicture}
\def \fsize{\Huge}


\def\nodewidth{0.4in}
\def\cx{-1}
\tikzstyle{bit} = [circle, draw, thick,fill=gray,opacity=0.7,line width =2, minimum size=\nodewidth]
\tikzstyle{check} = [rectangle, draw, thick,fill=white, line width =2, minimum size=\nodewidth]

\node[bit](v8) at  (-7,-10) {};
\node[bit](v1) at (-7,-7) {};
\node[bit] (v6) at (-7,-4){};
\node[bit] (v5) at (-7,-1) {};
\node[bit] (v3) at (-7,2) {};
\node[bit] (v1_2) at (-7,5) {} ;

\node [anchor=east] at (v1_2.west){\Huge $r[0]$};
\node [anchor=east] at (v3.west){\Huge $r[1]$};
\node [anchor=east] at (v5.west){\Huge $r[2]$};
\node [anchor=east] at (v6.west){\Huge $r[3]$};
\node [anchor=east] at (v1.west){\Huge $r[4]$};
\node [anchor=east] at (v8.west){\Huge $r[5]$};

\node [check](c1) at  (\cx,-8.5) {};
\node [check](c2) at  (\cx,-5.5) {};
\node [check](c3) at  (\cx,-2.5){};
\node [check](c4) at   (\cx,0.5) {};
\node [check](c5) at (\cx,3.5){};

\draw [dashed,line width=1.5pt] (c3)++(-0.5in,-0.5in)-- ++(3,0);

\node [anchor=west] at (c5.east) {\Huge $ \zv_{1,1} = \begin{bmatrix}
	 r_{1,1}[1] &=& r[0] &+& r[3]  \\
	 r_{1,2}[1] &=& r[0] w_2^{0} &+& r[3] w_2^{3} \\
	 \end{bmatrix}$};

\node [anchor=west] at (c4.east) {\Huge $  \zv_{1,2} = \begin{bmatrix}
	 r_{1,1}[2] &=& r[1] &+& r[4] \\
	 r_{1,2}[2] &=& r[1] w_2^{1} &+& r[4] w_2^{4}\\
	 \end{bmatrix}$};

\node [anchor=west] at (c3.east) {\Huge $  \zv_{1,3} = \begin{bmatrix}
	 r_{1,1}[3] &=& r[2] &+& r[5]  \\
	 r_{1,2}[3] &=& r[2] w_2^{2} &+& r[5] w_2^{5} \\
	 \end{bmatrix}$};

\node [anchor=west] at (c2.east) {\Huge $  \zv_{2,1} = \begin{bmatrix}
	 r_{2,1}[1]&=& r[0]  &+& r[2]  &+& r[4]  \\
	 r_{2,2}[1] &=& r[0] w_2^{0} &+& r[2] w_2^{2} &+& r[4] w_2^{4}\\
	 \end{bmatrix}$};

\node [anchor=west] at (c1.east){\Huge $  \zv_{2,2} = \begin{bmatrix}
	 r_{2,1}[2] &=&  r[1] &+& r[3] &+& r[5] \\ \vspace{3pt}
	 r_{2,2}[2] &=&  r[1] w_2^{1} &+& r[3] w_2^{3} &+& r[5] w_2^{5} \\
	 \end{bmatrix}$};

\draw[thick,line width =2]  (v1_2) edge (c4.west);
\draw[thick,line width =2]  (v1_2) edge (c2.west);

\draw[thick,line width =2]  (v3) edge (c1.west);
\draw[thick,line width =2]  (v5.east) edge (c3.west);
\draw[thick,line width =2]  (v3) edge (c5.west);
\draw[thick,line width =2]  (v5.east) edge (c2.west);

\draw[thick,line width =2]  (v6.east) edge (c4.west);
\draw[thick,line width =2]  (v1) edge (c5.west);
\draw[thick,line width =2]  (v8.east) edge (c3.west);
\draw[thick,line width =2]  (v6.east) edge (c1.west);
\draw [thick,line width =2] (v8.east) edge (c1.west);
\draw [thick,line width =2] (v1) edge (c2.west);
\node at (4,7) { \Huge \Huge $w_i^k := \left (  e^{j \frac{2 \pi s_i}{N}}\right )^{k}$};

\end{tikzpicture}}	
	\end{center}	
	\caption{Example of a Tanner graph formed in a RSIDFT framework with system parameters being $N=6$, $f_1=2, f_2=3$ (i.e., $d=2$) and $B=2$. The variable nodes (gray circles) represent the cross-correlation vector $\rv$ and the bin nodes (white squares) represent the binned observation vector $\zv_{i,k}$. The figure illustrates the relationship between $\zv_{i,k}$ and $\rv$.}\label{fig:factorgraph}
	\vspace{5 pt}
\end{figure}
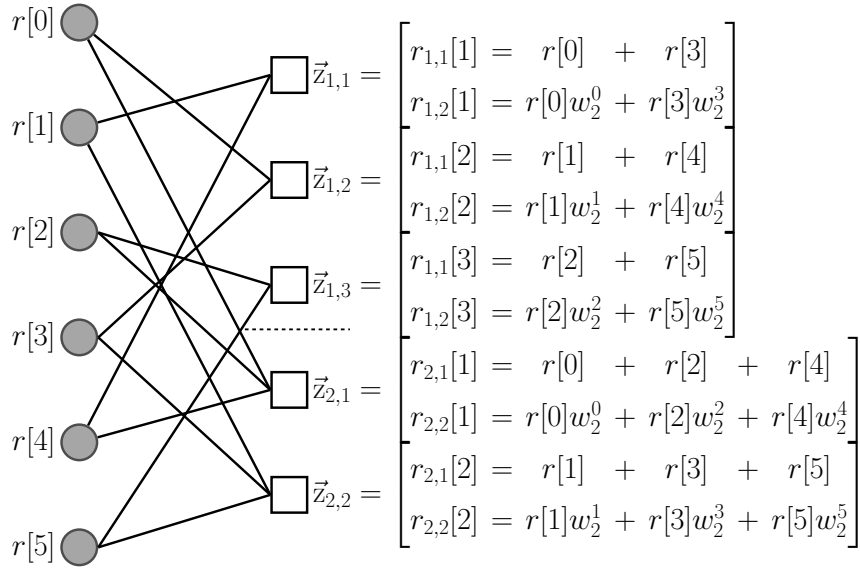
 
We represent the relation between the set of observation vectors $\{\zv_{i,k},i\in[1:d],k\in[f_i]\}$ and $\rv$ using a Tanner graph, an example of which is shown in Fig.~\ref{fig:factorgraph}. The nodes on the left, which we refer to as {\it variable nodes}, represent the $N$ elements of vector $\rv$. Similarly the nodes on the right, which we refer to as {\it bin nodes}, represent the $\sum_{i\leq d} f_i$ sub-sensing signals. We will now describe the decoding algorithm which takes the set of observation vectors $\{\zv_{i,k},i\in[1:d],k\in[f_i]\}$, each of length $B$, at $df_i$ bins as input and estimates the $L$-sparse $\rv$.	

\subsubsection{Decoder}			
	Observe from the Tanner graph that the degree of each variable node is $d$ and that of each bin node at stage $i$ is $g_i$. A variable node is referred to as non-zero if it corresponds to a matching position and as zero if it corresponds to a non-matching position. Note that even though the cross-correlation vector value corresponding to a non-matching position is not exactly zero but some negligible noise value we refer to them as zero variable nodes for simplicity. We refer to a bin node as {\it zero-ton} (or $\msc{H}_z$) if the number of non-zero variable nodes connected to the bin node is zero. The {\it singleton ($\msc{H}_s$), double-ton ($\msc{H}_d$) and multi-ton ($\msc{H}_m$)} bin nodes are defined similarly where the number of non-zero variable nodes connected are one, two and greater than two, respectively. The peeling decoder has the following three steps in the decoding process.

\paragraph*{Bin Classification} In this step a bin node is classified either as a zero-ton or a singleton or a multi-ton. At bin $(i,j)$ the classification is done based on  comparing the first observation $z_{i,j}[1]$, which corresponds to zero shift, with a predefined threshold. For $z_{i,j}[1]=z$, the classification hypothesis at bin $(i,j)$, $\widehat{\msc{H}}_{i,j}$, can be written as follows:
\begin{align}
\label{Eqn:BinClassifApprox}
\widehat{\msc{H}}_{i,j}=
\begin{cases}
\msc{H}_z &  	 z/M < \gamma_1\\
\msc{H}_s &	  \gamma_1 < z/M < \gamma_2  \\
\msc{H}_d  &    \gamma_2  < z/M <  \gamma_3\\
\msc{H}_m &      z/M > \gamma_3\\
\end{cases}
\end{align}
where $(\gamma_1,\gamma_2,\gamma_3)=(\frac{1-2\eta}{2},\frac{3-4\eta}{2},\frac{5-6\eta}{2})$. Note that for the case of exact matching $\eta=K/M=0$.
\paragraph*{Singleton decoding}
If a bin node $(i,j)$ is classified as a singleton in the bin classification step, we need to identify the position of the non-zero variable node connected to it. This is done by correlating the observation vector $\zv_{i,j}$ with each column of the sensing matrix  $\mb{W}_{i,j} \coleq [\wv^{j},\wv^{j+f_i},\ldots,   \wv^{j+(g_i-1)f_i}]$ and choosing the index that maximizes the correlation value.
\begin{align*}
 \hat{k} = \underset{k\in\{j+l f_i\}}{\argmax}~~ \zv^{\dagger}_{i,j} \wv^{k}
\end{align*}
where $\dagger$ denotes the conjugate transpose. The value of the variable node connected to the singleton bin is decoded as:
 $$
 \hat{r}[\hat{k}]=M(1-\eta).
 $$\vspace{-3pt}
 Note that for the case of exact matching we know the value to be exactly equal to $M$. But in the case of approximate matching, the actual value of $r[k]\in\{M(1-2\eta),\ldots,M-1,M\}$ and our estimate $\hat{r}[\hat{k}]=M(1-\eta)$ is only approximate. But this suffices for recovering the positions of matches i.e., the indices of the sparse coefficients in $\rv$.			

\paragraph*{Peeling Process} The peeling based decoder we employ consists of finding a singleton bin, then identifying the single non-zero variable node connected to the bin, decoding its value and removing ({\it peeling off}) it's contribution from all the bin nodes connected to that variable node. The main idea behind this decoding scheme is that (for appropriately chosen parameters), at each iteration, peeling a singleton node off will induce at least one more singleton bin and the process of peeling off can be repeated iteratively. Although the main idea is similar for exact matching and the approximate matching scenarios, there are some subtle differences in their implementation.\\
{\it Exact Matching}: In the case of exact matching, we remove the decoded variable node's contribution from all the bin nodes it is connected to.\\
{\it Approximate Matching}: In this case, similar to the approach in \cite{lee2015saffron},  we remove the decoded variable node's contribution only from bins that are originally a singleton or a double-ton. We do not alter the bins which are classified to be multi-tons with degree more than two.

{\it Note:} The differences in peeling implementation for exact and approximate matching cases is because unlike exact matching the approximate matching may result in non-zero error ($e_1$) between the actual correlation value $r[k]$ and the estimate $\hat{r}[\hat{k}]$, i.e. $e_1 = |r[k]- \hat{r}[\hat{k}]|$ with $0 \leq e_1 \leq \eta M$. This may lead to error propagation if we use the same decoding scheme as in exact matching. Hence to overcome this problem we impose a constraint on the type of bin nodes participating in the peeling process.
	
We present the overall recovery algorithm, which comprises of {\it bin classification, singleton decoding} and {\it peeling process}, in the Algorithm.\ref{Algo:decoder} pseudo-code. Note that $\msc{N}(k)$ denote the neighborhood for variable node $k$  i.e., the set of bins connected to $k^{\text{th}}$ variable node.

\def\gap{4pt}
\begin{algorithm}[h!]
\caption{Peeling based recovery algorithm}
\label{Algo:decoder}
\begin{algorithmic}
\State Compute $\widehat{\msc{H}}_{i,j}$ for $i\in[d], j\in[f_i]$. (See Eqn. \eqref{Eqn:BinClassifApprox})
\vspace{\gap}
\While {$\exists~ i,j:\widehat{\msc{H}}_{i,j}= \msc{H}_s$,}
\vspace{\gap}
  \State $(\hat{k},\hat{r}[\hat{k}])=${\bf Singleton-Decoder}$(\zv_{i,j})$
\vspace{\gap}
  \State Assign $\hat{r}[\hat{k}]$ to $\hat{k}^{\text{th}}$ variable node
\vspace{1.5\gap}
  \For{$(i_0,j_0)\in\msc{N}(\hat{k})$}
\vspace{\gap}
       \If {\it Exact Matching}
	   \State $\zv_{i_0,j_0}\gets\zv_{i_0,j_0}-\hat{r}[\hat{k}]\wv^{\hat{k}}\hspace{5.5ex} $ \hspace{20.5ex} 
	   \vspace{\gap}
        \Else
         
        \State $\zv_{i_0,j_0}\gets\zv_{i_0,j_0}-\hat{r}[\hat{k}]\wv^{\hat{k}}$   \hspace{2ex} only if $~~~~\widehat{\msc{H}}_{i_0,j_0}=\msc{H}_s$ or $\msc{H}_d$ 	   
        \hspace{14ex} 
        \EndIf
	   \vspace{\gap}
	   \State Re-do the bin classification for $(i_0,j_0)$ and compute $\widehat{\msc{H}}_{i_0,j_0}$
      \EndFor
\EndWhile
\end{algorithmic}
\end{algorithm}

\begin{algorithm}[h!]
\caption{Singleton-Decoder}
\label{Algo:SingletonDecoder}
\begin{algorithmic}
\State{\bf Input:} $\zv_{i,j}$
\vspace{\gap}
\State{\bf Output:} $(\hat{k},\hat{r}[\hat{k}])$
\vspace{\gap}
\State Estimate singleton index to be $ \hat{k} = \underset{k\in\{j+l f_i\}}{\argmax}\  \zv^{\dagger}_{i,j} \wv^{k}$
\vspace{\gap}
  \State Estimate the value to be:$$ \hat{r}[\hat{k}]=
   \begin{cases}
   M & \text{ Exact Matching case}\\
  M-K & \text{ Approximate Matching case}
  \end{cases}
  $$
\end{algorithmic}
\end{algorithm}

\subsubsection{Choosing $f_i$ and $\alpha$ for various $\mu$}
\label{subsec:DesignParameters}

For a given value of $\mu$, we will describe how to choose the parameters $d$ and $f_i$. Find a factorization for signal length $N=\prod_{i=0}^{d-1} P_i$ such that the set of integers $\{P_0,P_1,\ldots,P_{d-1}\}$ are pairwise co-prime and all the $P_i$ are approximately equal. More precisely, let $P_i=\mathbf{F}+O(1) ~\forall i$ for some value $\mathbf{F}$. We can add zeros at the end of the vector $\xv$ and increase the length of the vector until we are able to find a factorization that satisfies this property.
\begin{itemize}
	\item  For $\mu<0.5$: Choose $f_i = N/P_i$. \\
	{\it Exact Matching:} Find $d\in \mbb{N}\backslash \{1,2\}$ such that $\mu\in(\frac{1}{d},\frac{1}{d-1}]$\\
	{\it Approximate Matching}: If $\mu\in(\frac{1}{8},\frac{1}{2})$, choose $d=8$. Else find $d\geq 8$ such that $\mu\in(\frac{1}{d},\frac{1}{d-1})$\\
    \item For $\mu>0.5$: Choose $f_i = P_i$\\
    {\it Exact Matching:} Find $d\in \mbb{N}\backslash \{1,2\}$ such that $\mu\in(1-\frac{1}{d-1},1-\frac{1}{d}]$\\
    {\it Approximate Matching}: If $\mu\in(\frac{1}{2},\frac{1}{8})$, choose $d=8$. Else find $d\geq 8$ such that $\mu\in(1-\frac{1}{d-1},1-\frac{1}{d})$\\    	 
\end{itemize}
Thus, for both the exact and approximate matching cases, for any $0<\mu<1$, we choose the down-sampling factors $f_i$ to be approximately equal to $N^{\alpha}$ where $\alpha > 1- \mu$.

\subsubsection{Distributed processing framework}
Given a database(or string) of length $N$, we divide the database into $G=N^\gamma$ blocks each of length $\tilde{N} = N/G$. Now each block can be processed independently (in parallel) using the RSIDFT framework with the new database length reduced from $N$ to $\tilde{N}$. This distributed framework has the following advantages
\begin{itemize}
	\item Firstly, this enables parallel computing and hence can be distributed across different workstations.
	\item Improves the sample and computational complexity by a constant factor.
	\item Sketch of the database needs to be computed only for a smaller block length and hence requires computation of only a shorter $\tilde{N}$ point FFT.
	\item Helps overcome implementation issues with memory and precision as the scale of the problem is reduced. 
\end{itemize}
\subsection{Sketches of $\Xv$ and $\Yv$}
\label{subsec:skteches}		
 As we have already seen in Sec.~\ref{subsec:RSIDFT} the RSIDFT framework requires the values of $\Rv(=\Xv\odot \Yv)$ only at indices $\mc{S}$ or in other words we need $\Xv$ and $\Yv$ only at the indices in set $\mc{S}$ of cardinality $dBf_i$. We assume that the sketch of $\xv$, $ \Xv[\mc{S}]= \{X[i],i\in \mc{S}\}$ is pre-computed and stored in a database.

\subsubsection*{Computing the sketch of $\yv$}: For every new query $\yv$, only $\{\Yv'[\mc{S}_{i,j}]$,  $i\in[d],j\in[B]\}$ needs to be computed where the subsets $\mc{S}_{i,j}$, defined in Eq. \eqref{Eqn:SamplingSets}, are
\begin{align}
	 \mathcal{S}_{i,j} \coleq \{\ s_j + k g_i: k\in[f_i]\}, \quad i\in[d], j\in[B]
\end{align}
of cardinality $f_i$. Naively, the FFT algorithm can be used to compute $N$-pt DFT of $\Yv'$ and the required subset of coefficients can be taken but this is inefficient and would be of $O(N \log N)$ complexity. Instead, we observe that $\Yv'[\mc{S}_{i,j}]$ is $\Yv'$ shifted by $s_j$ and sub-sampled by a factor of $g_i$. Thus for a given $(i,j)$ this corresponds to, in time domain, a point-wise multiplication by $[1,w_{s_j},w_{s_j}^2,\ldots,w_{s_j}^{N-1}]$ followed by \textit{folding} the signal into $g_i(=\frac{N}{f_i})$ signals each of length $f_i$ and adding them up resulting in a single length-$f_i$ signal denoted by $\yv'_{i,j}$. Formally the \textit{folding} operation can be described as follows:
\begin{equation*}
	\yv'_{i,j} = \sum \limits_{m = 0}^{g_i-1} \yv'[mf_i:(m+1)f_i-1] \odot [w_{s_j}^{mf_i},w_{s_j}^{mf_i+1},\ldots, w_{s_j}^{(m+1)f_i-1}],
\end{equation*}	
where, $w_{s_j}=e^{-\frac{j2\pi s_j}{N}}$. Taking $f_i$-point DFT of $\yv'_{i,j}$ produces $\Yv'[\mc{S}_{i,j}]$ i.e. 
 \begin{align*}
	\Yv'[\mc{S}_{i,j}]=\mc{F}_{f_i} \{\yv'_{i,j}\}
\end{align*}
which is what we need in branch $(i,j)$. To obtain all the samples in $\mathcal{S}$ required for the RSDIFT framework, the \textit{folding} technique followed by a DFT needs to be carried out for each $(i,j)$, for $i\in[d],j\in[B]$, a total of $dB$ times $N^{\alpha}$-point DFT.

\section{Performance Analysis}
\label{sec:analysis}
\def\vgap{2pt}
In this section, we will analyze the overall probability of error involved in finding the correct matching positions. This can be done by analyzing the following three error events independently and then using a union bound to bound the total probability of error.

\begin{itemize}
	\item $\mathcal{E}_1${-\it Bin Classification}: Event that a bin is wrongly classified.
	\item $\mathcal{E}_2${-\it Position Identification}: Given a bin is correctly identified as a singleton, event that the position of singleton is identified incorrectly.
	\item $\mathcal{E}_3${-\it Peeling Process}: Given the classification of all the bins and the position identification of singletons in each iteration is accurate, event that the peeling process fails to recover the $L$ significant correlation coefficients.
\end{itemize}

\subsection{\bf Bin Classification}
\begin{lemma}
The probability of bin classification error at any bin $(i,j)$ can be upper bounded by
\begin{align*}
\mbb{P}[\mc{E}_1]\leq 6e^{-\frac{N^{\mu+\alpha-1}(1-6\eta)^2}{16}}
\end{align*} \label{Lem:binclassification}
\end{lemma}
\begin{proof}
\begin{align*}
\mbb{P}[\mc{E}_1] =& \leq \mbb{P}[\mc{E}_1|\widehat{\msc{H}}_{i,j}=\msc{H}_z]~+
						\quad \mbb{P}[\mc{E}_1|\widehat{\msc{H}}_{i,j}=\msc{H}_s]~ \quad + \quad \mbb{P}[\mc{E}_1|\widehat{\msc{H}}_{i,j}=\msc{H}_d \cup \msc{H}_m]\\
    			&\leq  e^{-\frac{N^{\mu-\alpha}(1-2\eta)^2}{8}}+2e^{-\frac{N^{\mu-\alpha}(1-4\eta)^2}{16}} \quad + 2e^{-\frac{N^{\mu+\alpha-1}(1-6\eta)^2}{16}}+e^{-\frac{N^{\mu+\alpha-1}(1-6\eta)^2}{16}}\\
    			&\leq 6e^{-\frac{N^{\mu+\alpha-1}(1-6\eta)^2}{16}}\\
 \end{align*}
{where the inequalities in the third line are due to Lemmas \ref{Lem:ZerotonClassif}, \ref{Lem:SingletonClassif}, \ref{Lem:DoubletonClassif} and \ref{Lem:MultitonClassif} respectively provided in Appendix \ref{Append:BinClassif}.}
\end{proof}

\subsection{\bf Position Identification}
\begin{lemma}
Given that a bin $(i,j)$ is correctly classified as a singleton, the probability of error in identifying the position of the non-zero variable node can be upper bounded by
\begin{align*}
\mbb{P}[\mc{E}_2]\leq & \exp\left\lbrace-\frac{N^{\mu+\alpha-1}(1-2\eta)^2(c_1^2-1)}{8(c_1^2+1)}\right\rbrace
 \quad + \quad \exp\left\lbrace-\frac{N^{\mu+\alpha-1} ~ (c_1(1 - 2\eta) - 1)^2}{8(1+ c_1^2)}\right\rbrace
\end{align*}\label{Lem:posidentification}
\end{lemma}
\begin{proof}
	The detailed proof is provided in Appendix \ref{Append:PositionIdentif}.
\end{proof}

\subsection{\bf Peeling Process}
\begin{table}[ht]
	\centering
	\begin{tabular}{ c  c  c  c  c  c  c  c  }
		\hline
		$d$ & $2$& $3$ & $4$ & $5$ & $6$ & $7$ & $8$ \\ \hline
		$\delta$ & 1.000 & 0.4073 & 0.3237 & 0.2850 & 0.2616 & 0.2456 & 0.2336 \\ 
		$d\delta$ & 2.000 & 1.2219 & 1.2948 & 1.4250 & 1.5696 & 1.7192 & 1.8688 \\ \hline
	\end{tabular}
	\vspace{1ex}
	\caption{Constants for various error floor values}
	\label{Table:EtaValues}
\end{table}
To analyze the peeling process alone independently, we refer to a {\it oracle based peeling decoder} which has the accurate classification of all the bins and can accurately  identify the position of the singleton given a singleton bin at any iteration. In other words, oracle based peeling decoder is the peeling part of our decoding scheme but with the assumption that the bin classification and position identification are carried out without any error.
\begin{lemma}
[Exact Matching]
For the exact matching case, choose $F^{d-1}=\delta N^\alpha$ where $\delta$ is chosen as given in Table. \ref{Table:EtaValues}. Then the oracle based peeling decoder:
\begin{itemize}
\item successfully uncovers all the $L$ matching positions if $L=\Omega(N^{\alpha})$ and $L\leq N^{\alpha}$, with probability at least $1-O(1/N^{\frac{1}{d}})$
\item successfully uncovers all the $L$ matching positions, if $L=o(N^{\alpha})$, with probability at least $1-e^{-\beta \varepsilon_1^2N^{\alpha/(4l+1)}}$ for some constants $\beta,\varepsilon_1>0$ and $l>0.$
\end{itemize}\label{Lem:peeling_exact}
\end{lemma}
\begin{proof}
We borrow this result from Pawar and Ramchandran's \cite{pawar2014robust}. Although our RSDIFT framework and their robust-FFAST scheme have three main differences:
\begin{itemize}
\item We are computing smaller IDFT's to recover a sparse bigger IDFT whereas in \cite{pawar2014robust} the same is true for DFT instead of IDFT.
\item Our problem model is such that the sparse components of the signal space has only positive amplitude and thus our bin processing part (bin classification and position identification) is different when compared to \cite{pawar2014robust}.
\item The sparsity of the signal $L$ to be recovered is exactly known in the case of \cite{pawar2014robust} whereas we have no information  about $L$ not even the order with which the quantity scales in $N$.
\end{itemize}

Irrespective of these differences, the Tanner graph representation of the framework and the peeling part of the decoder are identical to that of the robust-FFAST scheme. And thus the limit of the {\it oracle based peeling decoder} for our scheme is identical to that in the robust-FFAST scheme \cite{pawar2014robust}.  With respect to the third difference, in robust-FFAST scheme the authors choose $F^{d-1}=\delta k$ where $k$ is the sparsity of the signal (which is assumed to be known) and show the first assertion of the lemma. They also showed that upto a constant fraction $(1-\varepsilon)$ of $k$-variables node can be recovered with probability of failure that decays exponentially in $N$. In our case, since $L=o(N^{\alpha})$, this result translates to recovering all the $L$ non-zero variable nodes with an exponentially decaying failure  probability.
\end{proof}

In any iteration, given a singleton bin, the peeling process, in the case of approximate matching, runs the Singleton-Decoder algorithm on the bin only if it was either originally a singleton or originally a double-ton with one of the variable nodes being peeled off already. This is in contrast to the exact matching case where the peeling decoder runs the Singleton-Decoder on the bin irrespective of it's original degree. Hence we need to analyze the oracle based peeling decoder for the approximate matching case separately compared to the exact matching case.
\begin{lemma}
[Approximate Matching]
For the approximate matching case, choose parameter $d\geq 8$ as described in Sec. \ref{subsec:DesignParameters} and $F^{d-1}=0.7663 N^{\alpha}$. Then the oracle based peeling decoder: 
\begin{itemize}
\item successfully uncovers all but a small fraction $\varepsilon=10^{-3}$ of the $L$ matching positions, if $L=\Omega(N^{\alpha})$ and $L\leq N^{\alpha}$ with a failure probability that decays exponentially in $N$
\item successfully uncovers all the $L$ matching positions, if $L=o(N^{\alpha})$, with probability at least $1-e^{-\beta \varepsilon_1^2N^{\alpha/(4l+1)}}$ for some constants $\beta,\varepsilon_1>0$ and $l>0.$
\end{itemize}\label{Lem:peeling_approximate}
\end{lemma}
\begin{proof}
As mentioned earlier, the key difference in the approximate matching case is peeling off variable nodes from only singleton and double-tons. An identical peeling decoder is used and analyzed in the problem of group testing \cite{lee2015saffron} by Lee, Pedarsani and Ramchandran which the authors refer to as SAFFRON scheme. In SAFFRON, the authors claim that for a graph ensemble which has a regular degree of $d$ on the variable nodes and a  Poisson degree distribution on the bins, this peeling decoder with a left degree of $d=8$ and a total number of bins at least equal to $6.13 k$ recovers at least $(1-\varepsilon)$ fraction of the $k$ non-zero variable nodes with exponentially decaying probability. Note that $\frac{6.13}{8}\approx 0.7663$ is approximately the number of bins per stage for $d=8$. We also leverage the result from \cite{pawar2014robust} that the Tanner graph representation of the robust-FFAST (or equivaently RSDIFT framework) has a Poisson degree distribution on the bins. Combining these two results gives us the required results.
\end{proof}

\begin{theorem}[Overall Probability of Error] \label{thm:OverallErrorProb}
	The RSIDFT framework succeeds with high probability asymptotically in $N$ if the number of samples in each branch $f_i = N^{\alpha}+O(1)$ satisfies the condition $\alpha > 1-\mu$.
\end{theorem}
\begin{proof}
The overall probability of error $\mbb{P}[\mc{E}_{\text{total}}]$ can be bounded using an union bound on the three error events $\mc{E}_1$, $\mc{E}_2$ and $\mc{E}_3$ given by  	$$
\mbb{P}[\mc{E}_{\text{total}}] ~\leq~  \mbb{P}[\mc{E}_1] +  \mbb{P}[\mc{E}_2] + \mbb{P}[\mc{E}_3]
$$
Using the expressions for error probabilities from Lemmas \ref{Lem:binclassification}, \ref{Lem:posidentification}, \ref{Lem:peeling_exact} and \ref{Lem:peeling_approximate}, we can see that all the terms vanish to zero as $N \rightarrow \infty$ if $\mu+\alpha-1 \geq 0$, i.e. $\alpha \geq 1-\mu$.
\end{proof}

\section{Sample and Computational Complexity}
\label{Sec:Complexity}
In this section, we will analyze the sketching complexity which is the  number of samples we access from the sketch of the signal $\xv$ stored in the database and the computational complexity as a function of the system parameters.

\subsection{\bf Sample Complexity}\label{subsec:SampleComplexity}
In each branch of the RSDIFT framework we down-sample the $N$ samples by a factor of $\frac{N}{f_i}$ to get $f_i\approx N^{\alpha}$ samples. We repeat this for a random shift in each branch for $B=O(\log N)$ branches in each stage thus resulting in a total of $O(N^{\alpha}\log N)$ samples per block per stage. We repeat this for $d = \frac{1}{1-\alpha}$ such stages resulting in a total of $dN^{\alpha}\log N$ samples per block. So, the total number of samples is given by
\begin{align*}
S&= O \left(dN^{\alpha}\log N\right) =   O(N^{1-\mu}\log N)
\end{align*}

\subsection{\bf Computational Complexity} \label{subsec:ComputationComplexity}

As described in Eq.~\ref{eqn:Rxy_fourier}, the computation of $\rv$ involves three steps:
\begin{enumerate}
	\item  Operation - \RNum{1}:
	Since we assume that the sketch of database $\xv$, $\mathcal{F}_{N}\{\xv\}$, is pre-computed, we do not include this in computational complexity.

	\item  Operation - \RNum{2}:
	As described in Sec.~\ref{subsec:skteches}, in each branch $(i,j)$, we use a folding based technique to compute the sketch of $\yv$, $\mathcal{F}_{N}\{\yv'\}$ at points in the set $\mc{S}_{i,j}$. The folding technique involves two steps: folding and adding (aliasing) which has a complexity of $O(M)$ computations , and computing $f_i$-point DFTs that takes $O(N^\alpha \log N^{\alpha})$ computations. So, for a total of $dB$ branches the number of computations in this step is given by
	
	\begin{align*}
	 C_{\RNum{2}} \ &= ~  dB ~
	( \underset{\text{Folding} }{\underbrace{N^{\mu}}} + \ \
	\underset{\text{Shorter FFTs} }{\underbrace{N^{\alpha} \ \log N^{\alpha}}} \ )\\
	&= ~O(\max(N^{1-\mu}\log^2 N ,N^{\mu}\log N)).
	\end{align*}
	{\textit{Note:}} Folding and adding, for each shift, involves adding $N^{1-\alpha}$ vectors of length $N^{\alpha}$. We know that the length of the query is $M =N^{\mu}$, i.e., the number of non-zero elements in $\yv$ (zero-padded version of the query) is $M$ and hence we only need to compute $M$ additions instead of length of the vector $N$.

	\item  Operation - \RNum{3}:
	 Computing $\mathcal{F}_{N}^{-1}\{ \xv \}$ involves two parts:\\ RSIFT framework and the decoder. The RSIFT framework involves computing smaller $f_i$ point IDFTs, which takes approximately $O(N^{\alpha} \log N^{\alpha})$ computations in each branch. For a total of $dB$ branches, we get a complexity of $O(dB N^{\alpha}$ $\log N^{\alpha})$. In the decoding process, the dominant computation is from position identification. Each position identification process involves correlating the observation vector of length $B$ with $\frac{N}{f_i} \approx N^{1-\alpha}$ column vectors, which amounts to $B N^{1-\alpha}$ computations. There will be a maximum of $dL$ such position identifications, which gives a complexity of$O(dLBN^{1-\alpha} )$. Now, plugging in $\alpha = 1-\mu$ (condition for vanishing probability of error) the total number of computations involved in this step, $C_{\RNum{1}}$, is given by
	\begin{align*}	
	C_{\RNum{3}} \ &=  {d B}  ~(\hspace{-8pt}
	\underset{\text{Shorter IFFTs /block/stage} }{\underbrace{ O(N^{\alpha}  \log N^{\alpha})}} \hspace{-3pt}+ \underset{\text{Correlations} }{\underbrace{ L~N^{1-\alpha}}})
	\\
	&=  O(\max(N^{1-\mu}\log^2 N ,N^{\mu+\lambda}\log N)).
	\end{align*}

\end{enumerate}

Thus, the total number of computations, $C = \max\{C_{\RNum{2}},C_{\RNum{3}}\} $, is given by
  \begin{align*}
  C ~ = O(\max\{N^{1-\mu}\log^2 N, N^{\mu+\lambda}\log N \})
  \end{align*}

\section{Simulation Results}

\begin{figure*}[ht]
		\begin{tabular}{cc}
			\subfloat[$M=10^5(\mu=0.41), \tilde{N}=10^7, ~G=10^{5}$]{
%
%
\begin{tikzpicture}

\begin{axis}[%
width=2.521in,
height=1.566in,
at={(0.758in,0.481in)},
scale only axis,
xmin=100,
xmax=800,
xlabel={Sample Gain},
ymode=log,
ymin=1e-05,
ymax=0.1,
yminorticks=true,
ylabel={Prob of Missing a Match},
axis background/.style={fill=white}
]
\addplot [color=red,solid,mark=*,mark options={solid},forget plot]
  table[row sep=crcr]{%
167.8088	1e-07\\
218.0519	0.0005965\\
274.9922	0.0006834\\
338.1745	0.0012\\
484.384	0.027\\
750.2035	0.047\\
};
\end{axis}
\end{tikzpicture}%
}&
			\subfloat[$M=10^3(\mu=0.25),~ \tilde{N}=10^6, ~G=10^{6}$]{
%
%
\begin{tikzpicture}

\begin{axis}[%
width=2.521in,
height=1.566in,
at={(0.758in,0.481in)},
scale only axis,
xmin=2,
xmax=18,
xlabel={Sample Gain},
ymode=log,
ymin=1e-05,
ymax=1,
yminorticks=true,
ylabel={Prob of Missing a Match},
axis background/.style={fill=white}
]
\addplot [color=red,solid,mark=*,mark options={solid},forget plot]
  table[row sep=crcr]{%
2.0985	4e-07\\
3.2978	0.004\\
3.9974	0.004\\
4.7636	0.006\\
5.9963	0.017\\
7.3622	0.038\\
9.4944	0.054\\
17.4904	0.197\\
};
\end{axis}
\end{tikzpicture}%
}
		\end{tabular}
		
		\caption{Plots of probability of missing a match vs. sample gain for exact matching of a query of length $M$ from a equiprobable  binary \{+1,-1\} sequence of length $N= 10^{12}$, divided into $G$ blocks each of length $\tilde{N}$. The substring was simulated to repeat in $L=10^6$($\lambda=0.5$) locations uniformly at random.} \label{Fig:Simulation Results}
\end{figure*}
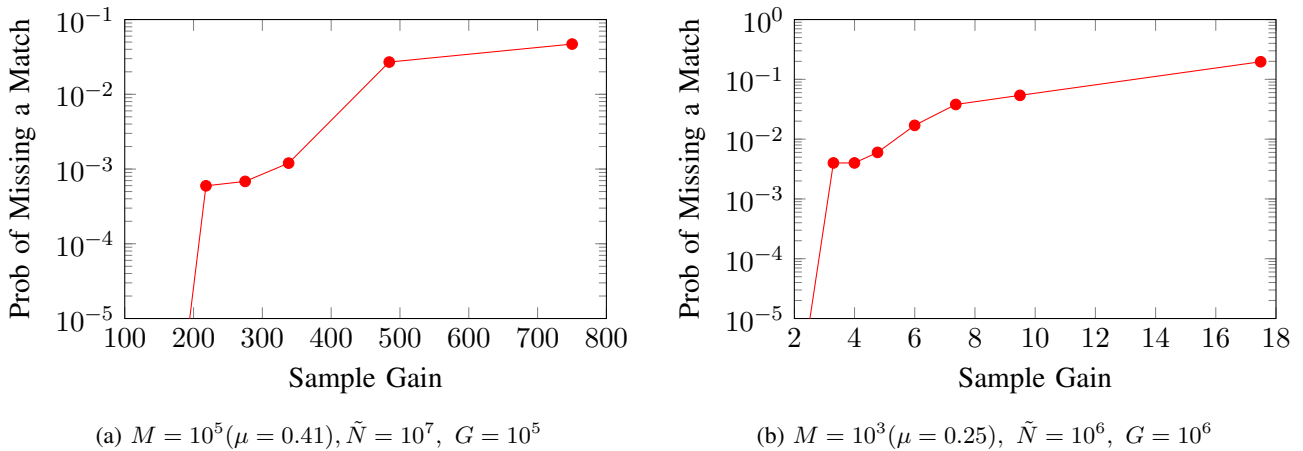

Simulations were carried out to test the performance of RSIDFT framework for exact matching scenario on a database of length $N=10^{12}$ for two different query lengths $M=10^5$ ($\mu = 0.41$) and $M=10^3$ ($\mu = 0.25$). The database was generated as a equiprobable $\{+1,-1\}$ sequence of length $N$. A substring of length $M$ from the generated database is presented as a query. Also the chosen query was repeated at $L=10^6$ randomly chosen locations in the database.

The sample gain, defined as the ratio of $N$ to the number of samples used from the sketch of database, was varied and the probability of RSIDFT framework to miss a match ($P_e$), as defined below, was measured.
\[P_e = \frac{\text{\# of correctly identified locations}}{L} \]   
The plots of $P_e$ vs. sample gain, is presented in Fig.~\ref{Fig:Simulation Results} for two different query lengths: $M=10^5~(\mu=0.41)$ in Fig~\ref{Fig:Simulation Results}(a) and $M=10^3~(\mu=0.25)$ in Fig~\ref{Fig:Simulation Results}(b). As can be inferred from the plots we achieve a sample gain of 200-300 (depending on the tolerable error probability) for the query length corresponding to  $\mu=0.41$ and a sample gain of $2$-$4$ for $\mu=0.25$. This sample gain results from an average number of samples per branch $f_i \approx 9.25 \times10^7 $ ($\alpha=0.66$) for $\mu=0.41$, and  $f_i \approx 6.94\times10^9 $ ($\alpha=0.82$) for $\mu=0.25$. The trend in the results almost matches with the theoretical findings of $\alpha = 1-\mu$. We also notice a sharp threshold in the sample gain, below which the RSIDFT framework succeeds with very high probability. 

\appendix
\section{Chernoff Bounds}

\begin{lemma}[Hoeffding tail bound]
\label{Lem:Chernoff}
Let $X_1, X_2,\ldots, X_n$ be a sequence of independent random variables such that $X_i$ has mean $\mu_i$ and sub-Gaussian parameter $\sigma_i$. Then for any $\delta>0$:
\begin{align*}
\text{\textbf{Upper Tail}}: ~&\mbb{P}\left[\sum \left(X_i-\mu_i\right)\geq \delta\right]\leq \exp\left\lbrace-\frac{\delta^2}{2\sum \sigma_i^2}\right\rbrace\\
\textbf{Lower Tail}: ~&\mbb{P}\left[\sum \left(X_i-\mu_i\right)\leq -\delta\right]\leq \exp\left\lbrace-\frac{\delta^2}{2\sum \sigma_i^2}\right\rbrace
\end{align*}
Note that for bounded random variables $X_i\in [a,b]$ the sub-Gaussian parameter is $\sigma_i=\frac{b-a}{2}$ whereupon the upper tail Hoeffding bound can be simplified to
\begin{align}
\mbb{P}\left[\sum_{i=1}^{n} \left(X_i-\mu_i\right)\geq \delta\right]\leq \exp\left\lbrace-\frac{2\delta^2}{n(b-a)^2}\right\rbrace.
\label{Eqn:HoeffdingBoundedRV}
\end{align}
Similarly the lower tail bound can be simplified.
\end{lemma}

\begin{lemma}[Tail bounds for noise terms]
	\label{Lem:tailbounds}
Let us consider $r[\theta_0],r[\theta_1],\ldots ,r[\theta_{g_i-1}]$ where $\theta_j=\theta_0+jf_i$ and $\theta_j \notin \{\tau_1,\ldots, \tau_L\}$ is not one of the matching positions for any $j$. Then for any $\delta>0$:\\
{\bf Upper Tail:}
\begin{align*}
\mbb{P}\left[ \left(\frac{1}{M}\sum\limits_{j\in[g_i]}\sum\limits_{k\in[M]}x[\theta_j+k]y[k]\right)\geq \delta\right]\leq \exp\left\lbrace-\frac{M\delta^2}{2g_i}\right\rbrace\\
\end{align*}
{\bf Lower Tail:}
\begin{align*}
\mbb{P}\left[ \left(\frac{1}{M}\sum\limits_{j\in[g_i]}\sum\limits_{k\in[M]}x[\theta_j+k]y[k]\right)\geq \delta\right]\leq \exp\left\lbrace-\frac{M\delta^2}{2g_i}\right\rbrace\\
\end{align*}
Recall that $[g_i]$ is used to denote the set $\{0,1,\ldots,g_i-1\}$.
\end{lemma}
\begin{proof}
Since $\theta_j$ is not one of the matching positions for any $j$, $x[\theta_j+k]\neq y[k]$ and more importantly $x[\theta_j+k]\indep y[k] ~\forall j,k$. This implies that $x[\theta_j+k]y[k]=\pm 1$ with equal probability and $\mbb{E}[x[\theta_i+k]y[k]]=0$. Let the set of random variables corresponding to a position $\theta_j$ be $S_{j}\coleq \{x[\theta_j+k]y[k],k\in[M]\}$. It is clear that the random variables in the set $S_j$ are all independent with respect to each other due to our i.i.d assumption on the database $\xv$ and $\theta_j$ being a non-matching position. For the case of $\mu<\alpha$ we have $M<f_i$ for large enough $N$ thus resulting in non-overlapping parts of $\xv$  participating in the correlation coefficients $r[\theta_i]$ and $r[\theta_j]$. Hence it can be shown that $S_i\indep S_j~~\forall i,j$ and we can apply the bounds from Eqn. \eqref{Eqn:HoeffdingBoundedRV} achieve the required result.

For the case of $\mu\geq\alpha$, $M>f_i$ for large enough $N$ which results in a coefficient $x[j]$ participating in multiple correlation coefficients $r[\theta_j]$. Therefore we pursue an alternate method of proof by defining 
$$
p_{j,l}\coleq x[\theta_j+l]\sum_{k\in[\frac{M}{f_i}]} y[l+kf_i] ~\text{ for } l\in[f_i],
$$
where $\theta_j=\theta_0+jf_i$. W.L.O.G we assume that $f_i$ divides $M$ evenly although the proof can be extended on similar lines for the case where $f_i$ does not divide $M$ evenly. Now we can show that the required sum
\begin{align*}
\sum\limits_{j\in[g_i]}\sum\limits_{l\in[M]}x[\theta_j+l]y[l]=\sum_{j\in[g_i]}\sum_{l\in[f_i]}p_{j,l}.
\end{align*}
From the above equivalent representation of the required sum, we need the following to achieve the required result:
\begin{itemize}
\item $p_{j,l}$ is sub-Gaussian with parameter $\sigma_i=\sqrt[•]{\frac{M}{f_i}}$ since, from Eqn. \eqref{Eqn:HoeffdingBoundedRV}, 
\begin{align*}
\mbb{P}\left[ p_{j,l}\leq \delta\right] \leq \exp\left\lbrace \frac{-\delta^2}{2\frac{M}{f_i}}\right\rbrace
\end{align*} 
\item The random variables $\{p_{j,l},~\forall j,l\}$ are independent of each other due to the i.i.d assumption on the database
\end{itemize}
 Now we can apply the tail bounds for sum of $g_i f_i$ sub-Gaussian random variables each with sub-Gaussian parameter $\sqrt{\frac{M}{f_i}}$ from Lem. \ref{Lem:Chernoff} and arrive at the required result.
\end{proof}


\section{Bin Classification Errors}
\label{Append:BinClassif}
We employ classification rules based only on the first element of the measurement vector at bin $(i,j)$ which can be given by
\begin{align}
Z[1]=\begin{cases}
\sum\limits_{\ell=0}^{g_{i}-1}\sum\limits_{k=0}^{M-1} n_{l,k}  & ~~\text{ if } ~~ \msc{H}=\msc{H}_z\label{Eqn:BinCombination}\\
\vspace{\vgap}
M_1+\sum\limits_{\ell=0}^{g_{i}-2}\sum\limits_{k=0}^{M-1} n_{l,k}  & ~~\text{ if } ~~ \msc{H}=\msc{H}_s\\
\vspace{\vgap}
M_1+M_2+\sum\limits_{\ell=0}^{g_{i}-3}\sum\limits_{k=0}^{M-1} n_{l,k}  & ~~\text{ if } ~~ \msc{H}=\msc{H}_d\\
\end{cases}
\end{align}
where $n_{l,k}=x[\theta_{\ell}+k]y[k]$ and $\theta_{\ell}\notin\{\tau_1,\tau_2,\ldots,\tau_L\}$. Also for the case of exact matching $M_1=M_2=M$ whereas in the case of approximate matching the values of $M_1,M_2\in[M(1-2\eta):M]$.

\begin{lemma}[zero-ton]
\label{Lem:ZerotonClassif}
Given that the bin $(i,j)$ is a zero-ton, the classification error can be bounded by
\begin{align*}
\mbb{P}[\mc{E}_1|\msc{H}_z]\leq e^{-\frac{N^{\mu+\alpha-1}(1-2\eta)^2}{8}}
\end{align*}
\end{lemma}
\begin{proof}
The above expression can be derived by observing that a bin is not classified as zero-ton if $\frac{Z[1]}{M}\geq\frac{1-2\eta}{2}$. Let us denote the probability of this event as $p_{z1}$ which can be bounded as:
\begin{align*}
p_{z1}=&\mbb{P}\left[\frac{Z[1]}{M}\geq\frac{1-2\eta}{2}\right]\\
&\leq e^{-\frac{Mg_i(1-2\eta)^2}{8g_i^{2}}}\\
&\approx e^{-\frac{N^{\mu+\alpha-1}(1-2\eta)^2}{8}}
\end{align*} 
where the second bound is due to Eqn. \eqref{Eqn:BinCombination} and Lemma.~\ref{Lem:tailbounds}. The approximation in the third line is from our design that all the $g_i$ are chosen such that $g_i\approx N^{1-\alpha}$ and $M=N^{\mu}.$
\end{proof}

\begin{lemma}[singleton]
\label{Lem:SingletonClassif}
Given that the bin $(i,j)$ is a singleton, the classification error can be bounded by
\begin{align*}
\mbb{P}[\mc{E}_1|\msc{H}_s]\leq 2e^{-\frac{N^{\mu+\alpha-1}(1-4\eta)^2}{16}}
\end{align*}
\end{lemma}
\begin{proof}
We observe that a bin is not classified as singleton if $\frac{Z[1]}{M}\leq\frac{1-2\eta}{2}$ or $\frac{Z[1]}{M}\geq\frac{3-4\eta}{2}$. Let us denote the probability of the two events as $p_{s1}$ and $p_{s2}$ respectively which can be bounded as:
\begin{align*}
p_{s1}&=\mbb{P}\left[\frac{1}{M}\sum\limits_{\ell=0}^{g_{i}-2}\sum\limits_{k=0}^{M-1} n_{l,k}\leq\frac{1-2\eta}{2}-\frac{M_1}{M}\right]\\
&\leq \mbb{P}\left[\frac{1}{M}\sum\limits_{\ell=0}^{g_{i}-2}\sum\limits_{k=0}^{M-1} n_{l,k}\leq-\frac{1-2\eta}{2}\right]\\
&\leq e^{-\frac{Mg_i(1-2\eta)^2}{16g_i^{2}}}\\
&\approx e^{-\frac{N^{\mu+\alpha-1}(1-2\eta)^2}{16}}
\end{align*} 
where we used  {\it lower tail} of Lemma \ref{Lem:Chernoff} and $g_i\approx N^{1-\alpha}$ and the lower bound on $\frac{M_1•}{M}\geq (1-2\eta)$. Similarly $p_{s2}$ can be upper bounded by:
\begin{align*}
p_{s2}&=\mbb{P}\left[\frac{1}{M}\sum\limits_{\ell=0}^{g_{i}-2}\sum\limits_{k=0}^{M-1} n_{l,k}\geq\frac{3-4\eta}{2}-\frac{M_1}{M}\right]\\
&\leq \mbb{P}\left[\frac{1}{Mg_i}\sum\limits_{\ell=0}^{g_{i}-2}\sum\limits_{k=0}^{M-1} n_{l,k}\geq-\frac{1-4\eta}{2g_i}\right]\\
&\approx e^{-\frac{N^{\mu+\alpha-1}(1-4\eta)^2}{8}}
\end{align*} 
Thus the overall probability of error for classifying a singleton can be obtained by combining $p_{s1}$ and $p_{s2}$.
\end{proof}

\begin{lemma}[double-ton]
\label{Lem:DoubletonClassif}
Given that the bin $(i,j)$ is a double-ton, the classification error can be bounded by
\begin{align*}
\mbb{P}[\mc{E}_1|\msc{H}_d]\leq 2e^{-\frac{N^{\mu+\alpha-1}(1-6\eta)^2}{16}}
\end{align*}
\end{lemma}
\begin{proof}
We observe that a bin is not classified as double-ton if $\frac{Z[1]}{M}\leq\frac{3-4\eta}{2}$ or $\frac{Z[1]}{M}\geq\frac{5-6\eta}{2}$. Let us denote the probability of these two events as $p_{d1}$ and $p_{d2}$ respectively which can be bounded similar to Lemma ~\ref{Lem:SingletonClassif}.
\begin{align*}
p_{d1}&=\mbb{P}\left[\frac{1}{M}\sum\limits_{\ell=0}^{g_{i}-3}\sum\limits_{k=0}^{M-1} n_{l,k}\leq\frac{3-4\eta}{2}-\frac{M_1+M_2}{M}\right]\\
&\leq \mbb{P}\left[\frac{1}{M}\sum\limits_{\ell=0}^{g_{i}-3}\sum\limits_{k=0}^{M-1} n_{l,k}\leq-\frac{1-4\eta}{2}\right]\\
&\leq e^{-\frac{M(g_i-2)(1-4\eta)^2}{16(g_i-2)^{2}}}\\
&\approx e^{-\frac{N^{\mu+\alpha-1}(1-4\eta)^2}{16}}.
\end{align*} 
Similarly $p_{d2}$ can be bounded as 
\begin{align*}
p_{d2}&=\mbb{P}\left[\frac{1}{M}\sum\limits_{\ell=0}^{g_{i}-3}\sum\limits_{k=0}^{M-1} n_{l,k}\geq\frac{5-6\eta}{2}-\frac{M_1+M_2}{M}\right]\\
&\leq \mbb{P}\left[\frac{1}{M}\sum\limits_{\ell=0}^{g_{i}-3}\sum\limits_{k=0}^{M-1} n_{l,k}\geq\frac{1-6\eta}{2}\right]\\
&\leq e^{-\frac{M(g_i-2)(1-6\eta)^2}{8(g_i-2)^{2}}}\\
&\approx e^{-\frac{N^{\mu+\alpha-1}(1-6\eta)^2}{8}}
\end{align*} 
where we use the lower bounds $M_1,M_2\leq M$.
\end{proof}

\begin{lemma}[multi-ton]
\label{Lem:MultitonClassif}
Given that the bin $(i,j)$ is a multi-ton, the classification error can be bounded by
\begin{align*}
\mbb{P}[\mc{E}_1|\msc{H}_m]\leq e^{-\frac{N^{\mu+\alpha-1}(1-6\eta)^2}{16}}
\end{align*}
\end{lemma}
\begin{proof}
We observe that a bin is not classified as multi-ton if $\frac{Z[1]}{M}\leq\frac{5-6\eta}{2}$. Let us denote the probability of this event as $p_{m1}$ which can be bounded as:
\begin{align*}
p_{m1}&=\mbb{P}\left[\frac{1}{M}\sum\limits_{\ell=0}^{g_{i}-3}\sum\limits_{k=0}^{M-1} n_{l,k}\leq\frac{5-6\eta}{2}-\frac{M_m}{M}\right]\\
&\leq \mbb{P}\left[\frac{1}{M}\sum\limits_{\ell=0}^{g_{i}-m}\sum\limits_{k=0}^{M-1} n_{l,k}\leq-\frac{1-6\eta}{2}\right]\\
&\leq e^{-\frac{M(g_i-m)(1-4\eta)^2}{16(g_i-m)^{2}}}\\
&\leq e^{-\frac{M(1-6\eta)^2}{16 n_i}}\\
 &\approx e^{-\frac{N^{\mu+\alpha-1}(1-6\eta)^2}{16}}.
\end{align*} 
\end{proof}

\section{Position Identification}
\label{Append:PositionIdentif}
We will analyze the singleton identification in two separate cases:
\begin{itemize}
\item $\mc{E}_{21}$: Event where the position is identified incorrectly when the bin is classified  correctly a singleton
\item $\mc{E}_{22}$: In the case of approximate matching, event where the position is identified incorrectly when the bin is originally a double-ton and one of the non-zero variable nodes has already been peeled off
\end{itemize}

\begin{definition}[Mutual Incoherence]
	The mutual incoherence $\mu_{\text{max}}( \mb{W})$ of a matrix $\mb{W} = [\wv_1 ~ \wv_2 ~ \cdots \wv_i \cdots \wv_N ]$ is defined as 
	
	\[\mu_{\text{max}}(\mb{W}) \defeq \max \limits_{\forall i \neq j} \frac{|\wv_i^{\dagger} \wv_j |}{||\wv_i || . ||\wv_j ||} \]
\end{definition}

\begin{lemma}[Mutual Incoherence Bound for sub-sampled IDFT matrix  [\cite{pawar2014robust},Proposition~A.1]
\label{lemma:MutualCoherence}
	The mutual incoherence $\mu_{\text{max}}$ $(\mb{W_{i,k}})$ of the sensing matrix $\mb{W}_{i,k}$ (defined in Eq.~\ref{Eqn:Sensing Matrix}), with $B$ shifts, is upper bounded by
	
	\[ \mu_{\text{max}} < 2\sqrt{\frac{\log(5N)}{B}} \] 
	
\end{lemma}
\begin{proof}
	The proof follows similar lines as the proof for Lemma V.3. in \cite{pawar2014robust}.
\end{proof}
 
\begin{lemma}
For some constant $c_1 \in \mathbb{R}$ and the choice of $B=4c_1^2\log 5N$,  the probability of error in identifying the position of a singleton at any bin $(i,j)$ can be upper bounded by
\begin{align*}
\mbb{P}[\mc{E}_{21}]\leq \exp\left\lbrace-\frac{N^{\mu+\alpha-1}(1-2\eta)^2(c_1^2-1)}{8(c_1^2+1)}\right\rbrace
\end{align*}
\end{lemma}
\begin{proof}
	
	Let $j_p$ be the variable node participating in the singleton $(i,j)$. Then the observation vector $\zv_{i,j}$ is given by
	\begin{align*}
	\underline{z}_{i,j} &= \begin{bmatrix}
	\wv_{j_{1}},\wv_{j_2}, & \cdots   & \wv_{j_p}, &\cdots \ &\wv_{j_{g_i}}
	\end{bmatrix} \times
	\begin{bmatrix}
	n_{1} \\
	\vdots \\
	r[j_p]\\
	\vdots\\
	n_{j} \\
	\vdots\\
	n_{g_i}\\
	\end{bmatrix}\\
	&= r[j_p] ~ \wv_{j_p}+ \sum_{k \neq p}n_k \wv_{j_k} \\
	\end{align*}
	where for convenience we use a simpler notation $j_k=j+(k-1)\frac{N}{f_i}, \wv_{j_k}=\wv^{j_k}$ as defined in Eq. and $n_{l}=\sum\limits_{k=0}^{M-1}x[\theta_{\ell}+k]y[k]$ as defined in Eq. \eqref{Eqn:BinCombination}.
	
	The estimated position $\hat{p}$ is given by
	\begin{align}
	\label{Eqn:SingletonBinCombination}
	\hat{p}= \underset{l}{\argmax}~~ \frac{\wv_{j_l}^{\dagger}\underline{z}_{i,j}}{B}
	\end{align}
	where $\dagger$ denotes the conjugate transpose of the vector. Also note that $|| \wv_{j_k}||=B$ for any $j$ and $k$.  From Eq. \eqref{Eqn:SingletonBinCombination} we observe that the position is wrongly identified when $\exists p'$ such that
	\begin{align*}
	&r[j_p] + \frac{1}{B}\sum_{k \neq p} n_k 	\wv_{j_p}^{\dagger}\wv_{j_k} \leq \frac{r[j_p]}{B} ~ \wv_{j_{p'}}^{\dagger}\wv_{j_p}+ n_{p'}+\frac{1}{B}\sum_{k \neq p,p'}n_k\wv_{j_{p'}}^{\dagger} \wv_{j_k} \\
	&\leftrightarrow \sum_{k \neq p,p'}\alpha_k n_k+\beta n_{p'}\geq  r[j_p]\left(1-\frac{\wv_{j_{p'}}^{\dagger}\wv_{j_p}}{B}\right)\geq M(1-2\eta)(1-\mu_{\text{max}})
	\end{align*}
	where $\alpha_k$ and $\beta$ are constants and can be shown to be in the range $\alpha_k\in[-2\mu_\text{max},2\mu_\text{max}]$ and $\beta\in[1-\mu_\text{max},1+\mu_\text{max}]$. Now using the bound given Chernoff Lemma in Lem.~\ref{Lem:tailbounds} we obtain
	\begin{align*}
	\mbb{P}[\mc{E}_{21}]&\leq \exp\left\lbrace-\frac{2M(1-2\eta)^2(1-\mu_{\text{max}})^2}{16g_i\mu^2_{\max}+4(1+\mu_{\max})^2}\right\rbrace\\
	&\leq\exp\left\lbrace-\frac{2M(1-2\eta)^2(1-\mu_{\text{max}})^2}{16(g_i\mu^2_{\max}+1)}\right\rbrace\\
	&\leq\exp\left\lbrace-\frac{2M(1-2\eta)^2(c_1-1)^2}{16(g_i+c_1^2)}\right\rbrace\\
	&\approx\exp\left\lbrace-\frac{N^{\mu+\alpha-1}(1-2\eta)^2(c_1^2-1)}{8(c_1^2+1)}\right\rbrace\\
	\end{align*}
	The second inequality follows by the definition of  $\mu_{\text{max}} \leq 1$.  We choose $B=4c_1^2\log 5N$, and substituting $\mu_{\max}\leq 2\sqrt{\frac{\log 5N}{B}} = 1/c_1$ (Lemma \ref{lemma:MutualCoherence}) we get the third inequality.
	
\end{proof}
\begin{lemma}
For some constant $c_1 \in \mathbb{R}$ and the choice of $B=4c_1^2\log 5N$, the probability of error in identifying the position of second non-zero variable node at a double-ton at any bin $(i,j)$, given that the first position identification is correct, can be upper bounded by
	\begin{align*}
		\mbb{P}[\mc{E}_{22}]\leq \exp\left\lbrace-\frac{N^{\mu+\alpha-1} ~ (c_1(1 - 2\eta) - 1)^2}{8(1+ c_1^2)}\right\rbrace
	\end{align*}
\end{lemma}
\begin{proof}
	
	{\bf $\mc{E}_{22}$:}
	
	Let $j_p$ and $j_{\tilde{p}}$ be the two variable nodes participating in the doubleton $(i,j)$. Then the observation vector $\zv_{i,j}$ is given by 
	
	\begin{align*}
		\underline{z}_{i,j} &= \begin{bmatrix}
			\wv_{j_{1}},\wv_{j_2}, & \cdots   & \wv_{j_p}, &\cdots \ &\wv_{j_{g_i}}
		\end{bmatrix} \times
		\begin{bmatrix}
			n_{1} \\
			\vdots \\
			r[j_p]\\
			\vdots\\
			n_{j} \\
			\vdots\\
			r[j_{\tilde{p}}]\\
			\vdots\\
			n_{g_i}\\
		\end{bmatrix}\\
		&= r[j_p] ~ \wv_{j_p} + r[j_{\tilde{p}}] ~ \wv_{j_{\tilde{p}}} + \sum_{k \neq p}n_k \wv_{j_k} \\
	\end{align*}
	
	Let the contribution from $j_{\tilde{p}}$ be peeled off from the doubleton at some iteration, then we get
	\[ \zv_{i,j} = r[j_p] ~ \wv_{j_p} + \frac{e_1}{B} ~ \wv_{j_{\tilde{p}}} + \sum_{k \neq p}n_k \wv_{j_k}\]
	
	where $e_1 \in[-\eta M, \eta M]$ is an extra error term induced due to peeling off.
	
	Now the estimated second position $\hat{p}$ is calculated using Eq. \eqref{Eqn:SingletonBinCombination}. We can observe that the position is wrongly identified when $\exists p'$ such that
	\[ \frac{\wv_{j_p}^{\dagger}\underline{z}_{i,j}}{B} \leq \frac{\wv_{j_{p'}}^{\dagger}\underline{z}_{i,j}}{B}\]
	\begin{align*}
		&\implies r[j_p] + \frac{1}{B}\sum_{k \neq p, \tilde{p}} n_k 	\wv_{j_p}^{\dagger}\wv_{j_k} + \frac{e_1}{B} \wv_{j_p}^{\dagger}\wv_{j_{\tilde{p}}} \\ & \qquad \leq \frac{r[j_p]}{B} ~ \wv_{j_{p'}}^{\dagger}\wv_{j_p}+ n_{p'}+\frac{1}{B}\sum_{k \neq p,p',\tilde{p}}n_k\wv_{j_{p'}}^{\dagger} \wv_{j_k} + \frac{e_1}{B} \wv_{j_{p'}}^{\dagger}\wv_{j_{\tilde{p}}}
	\end{align*}
	\begin{align*}
		&\leftrightarrow \sum_{k \neq p,p',\tilde{p}}\alpha_k n_k+ \beta n_{p'}  \geq  r[j_p]\left(1-\frac{\wv_{j_{p'}}^{\dagger}\wv_{j_p}}{B}\right) - \frac{2 \eta M}{B} \wv_{j_{p'}}^{\dagger}\wv_{j_{\tilde{p}}}
	\end{align*}
	\[~~\geq M(1-2\eta)(1-\mu_{\text{max}}) - 2 \eta M \mu_{\text{max}} = M(1 - 2\eta - \mu_{\text{max}})                         
	\]
	where $\alpha_k$ and $\beta$ are constants and can be shown to be in the range $\alpha_k\in[-2\mu_\text{max},2\mu_\text{max}]$ and $\beta\in[1-\mu_\text{max},1+\mu_\text{max}]$. Now using the bound given by Chernoff Lemma in Lem.~\ref{Lem:tailbounds} we obtain
	\begin{align*}
		\mbb{P}[\mc{E}_{22}]&\leq \exp\left\lbrace-\frac{2M(1 - 2\eta - \mu_{\text{max}})^2}{16g_i\mu^2_{\max}+4(1+\mu_{\max})^2}\right\rbrace\\
		&\leq\exp\left\lbrace-\frac{2M(1 - 2\eta - \mu_{\text{max}})^2}{16(g_i\mu^2_{\max}+1)}\right\rbrace\\
		&\leq\exp\left\lbrace-\frac{M(c_1(1 - 2\eta) - 1)^2}{8(g_i+c_1^2)}\right\rbrace\\
		&\leq\exp\left\lbrace-\frac{N^{\mu+\alpha-1} ~ (c_1(1 - 2\eta) - 1)^2}{8(1+ c_1^2)}\right\rbrace\\
	\end{align*}
	where for the choice of $B=4c_1^2\log 5N$, $\mu_{\max}\leq 2\sqrt{\frac{\log 5N}{B}} = 1/c_1$.
	
\end{proof}

\bibliographystyle{ieeetr}
\bibliography{journal_abbr,sparseestimation}

\end{document}